\documentclass[american,aps,pra,reprint, superscriptaddress]{revtex4-1}
\usepackage{amsmath,amscd,amsthm,amssymb}
\usepackage{mathrsfs}
\usepackage{graphicx}
\usepackage{epsf,epstopdf}
\usepackage{caption3}
\usepackage[all]{xy}
\usepackage{tikz}

\usepackage[unicode=true,pdfusetitle, bookmarks=true,bookmarksnumbered=false,bookmarksopen=false, breaklinks=false,pdfborder={0 0 0},backref=false,colorlinks=false] {hyperref}
\hypersetup{ colorlinks,linkcolor=myurlcolor,citecolor=myurlcolor,urlcolor=myurlcolor}
\usepackage{graphics,epstopdf,graphicx, amsthm, amsmath, amssymb, times, braket, color, bm}
\usepackage[up]{subfigure}
\usepackage{cleveref}

\definecolor{myurlcolor}{rgb}{0,0,0.7}

\def\be{\begin{equation}}
\def\ee{\end{equation}}
\def\bea{\begin{eqnarray*}}
\def\eea{\end{eqnarray*}}
\def\ot{\otimes}

\theoremstyle{plain}
\newtheorem{thrm}{\protect\theoremname}

\newtheorem{exam}[thrm]{Example}
\providecommand{\theoremname}{Theorem}

\newcommand{\iinner}[2]{\langle #1 | #2\rangle}
\newcommand{\out}[2]{| #1\rangle\langle #2 |}
\DeclareMathOperator{\trace}{tr}
\newcommand{\ptr}[2]{\trace_{#1}({#2})}
\newcommand{\tr}[1]{\ptr{}{#1}}
\newcommand{\id}{\mathbb{I}}

\newcommand*{\myproofname}{Proof}

\makeatother

\def\cC{\mathcal{C}}\def\cD{\mathcal{D}}
\def\cH{\mathcal{H}}\def\cI{\mathcal{I}}
\def\cO{\mathcal{O}}
\def\cS{\mathcal{S}}

\def\bE{\mathbf{E}}

\def\bP{\mathbf{P}}

\theoremstyle{definition}

\theoremstyle{remark}

\graphicspath{{figs/}}

\begin{document}

 \author{Sunho Kim}
 \email{kimsunho81@hrbeu.edu.cn}
 \affiliation{School of Mathematical Sciences, Harbin Engineering University, Harbin 150001, China}

  \author{Chunhe Xiong}
 \email{xiongchunhe@zju.edu.cn}
\affiliation{School of Computer and Computing Science, Zhejiang University City College, Hangzhou 310015, China}

\author{Asutosh Kumar}
 \email{asutoshk.phys@gmail.com}
 \affiliation{P.G. Department of Physics, Gaya College, Magadh University, Rampur, Gaya 823001, India}
 \affiliation{Harish-Chandra Research Institute, HBNI, Chhatnag Road, Jhunsi, Allahabad 211019, India}
 \affiliation{Vaidic and Modern Physics Research Centre, Bhagal Bhim, Bhinmal, Jalore 343029, India}

 \author{Junde Wu}
 \email{wjd@zju.edu.cn}
 \affiliation{School of Mathematical Sciences, Zhejiang University, Hangzhou 310027, China}

\title{Converting coherence based on positive-operator-valued measures into entanglement}
\begin{abstract}
Quantum resource theories provide a diverse and powerful framework for extensively studying the phenomena in quantum physics. Quantum coherence, a quantum resource,
is the basic ingredient in many quantum information tasks.
It is a subject of broad and current interest in quantum information, and many new concepts have been introduced and generalized since its establishment.
Here we show that the block coherence can be transformed into entanglement via a block incoherent operation. Moreover, we find that the coherence based on positive-operator-valued measures (POVMs) associated with block coherence through the Naimark extension acts as a potential resource from the perspective of generating entanglement. Finally, we discuss avenues of creating entanglement from POVM-based coherence, present strategies that require embedding channels and auxiliary systems, give some examples, and generalize them.

\end{abstract}
\maketitle

\section{Introduction}
Quantum systems outperform classical physical systems and exhibit exotic properties. Among these properties of quantum systems, certain properties have been recognised as vital resources for quantum information processing tasks. These realizations have spurred the development of general quantum resource theories (QRTs) \cite{Chitambar, Brandao}.

QRTs  provide a structured framework that quantitatively describes quantum properties such as entanglement \cite{Horodecki}, coherence \cite{Streltsov1, Streltsov2}, purity, quantum reference frames \cite{Bartlett, Gour}, quantum thermodynamics \cite{Goold}, contextuality \cite{Kochen}, nonlocality \cite{Bell, Brunner}, non-Gaussianity \cite{Weedbrook} and many others.
Each of these resources is classified as either classical or quantum, static (\emph{i.e.}, quantum states) or dynamic (\emph{i.e.}, quantum channels) and can be used for numerous interesting quantum information processing tasks.
Common ingredients of all QRTs include free states, free operations, and quantum resources. These components are not independent of each other in the sense that free operations alone cannot convert free states into resource states. Moreover, it is possible to asymptotically reverse many QRTs  if the set of free operations is maximal (i.e., consists of all possible operations that can not generate a resource from free states).

Coherence is a fundamental aspect of quantum theory \cite{Leggett}, felicitating the defining properties from superposition principle to quantum correlations. 
Several attempts have been made in recent years to characterize and quantify coherence. Baumgratz \emph{et al.} \cite{Baumgratz} provided a strict and clear framework for quantifying coherence similar to entanglement for general quantum systems. In the resource theory of quantum coherence, ``incoherent quantum states" are regarded as free states, and ``incoherent quantum channels" (operations which cannot generate coherence) are considered free operations that map a set of incoherent states to itself. Since then various measures of quantum coherence such as the K-coherence, robustness of coherence, distance-based coherence, relative entropy of coherence, the $l_1$-norm of coherence and partial coherence, etc \cite{Girolami, Lostaglio, Shao, Pires, Rana, Rastegin, Napoli, Yu, Luo, Bu, Xiong} have been introduced. 

The essence of both quantum correlations (including entanglement) and quantum coherence lies in the superposition principle. It means that the study of quantitative associations between them, namely how one resource is converted into another, is interesting and worthwhile.
Recently, an interesting framework for interconversion between coherence and entanglement, which provides a quantitative connection between them, was proposed by Streltsov \emph{et al.} \cite{Streltsov4}. Furthermore, it has been proved that quantum coherence can create other resources via incoherent operations \cite{Ma, Tan,  Mukhopadhyay} beyond quantum entanglement such as quantum Fisher information \cite{Kay, Genoni}, superradiance and quantum discord \cite{Ollivier}, etc.

Coherence, like entanglement, is a prerequisite for the advantage of quantum technologies. It provides quantum advantages in various tasks such as quantum thermodynamics \cite{Goold, Kammerlander}, quantum algorithms \cite{Hillery}, steering \cite{Mondal}, and quantum biology \cite{Huelga}. It is also beneficial for cooling in quantum refrigerators \cite{Mitchison} and extracting work using quantum thermal machines \cite{Kammerlander, Korzekwa}.

In the resource theory of typical quantum coherence, the free states are represented by diagonal matrices in fixed orthonormal bases. In other words, they can be considered to be caused by projective (von Neumann) measurements.
Thus, there is an inherent connection between coherence and quantum measurement. This notion was generalized by extending standard coherence to block-coherence based on general projective measurement \cite{Aberg}. 
It is well known that positive-operator-valued measures (POVMs) describe the most general quantum measurements (observables) and can provide a real operational advantage compared to any projective measurement (see Ref. \cite{Oszmaniec} for example). 
Hence, it is natural to develop the notion of coherenece with respect to general (nonprojective) quantum measurements.
The notion of POVM-coherence \cite{Bischof1} has been introduced very recently.  Using the rich structure of POVMs compared to projective measurements, POVM-based coherence was embedded into a resource-theoretic framework that generalizes the standard resource theory of coherence. The notion of POVM-based coherence requires Naimark extension, i.e., embedding the states and operations into a higher-dimensional space. Such a structure allows for a simpler derivation of general results, and describes the possibility of  directly implementing a POVM in an experiment.
In another recent work \cite{Bischof2} authors explored features of POVM-coherence that are distinct from standard coherence theory. In particular, they provided an operational interpretation to the concept of coherence with respect to general quantum measurements. They showed that {\it the relative entropy of POVM-coherence is equal to the cryptographic randomness gain}. 
%
Furthermore, in Ref. \cite{Ujjwal} authors found that an arbitrary set of quantum state vectors has a connection with POVMs. 
The results and findings in \cite{Bischof1, Bischof2, Ujjwal} are expected to empower coherence in quantum information technologies employing arbitrary nonprojective measurements.

We already know that entanglement plays a central role in quantum communication and computation \cite{Nielsen}, and coherence underlies entanglement. Therefore, it is natural to investigate if there is any connection between POVM-coherence and other quantum resources. In particular, if it is possible to convert POVM-coherence into entanglement.  
In this paper, we provide a connection of POVM-coherence with entanglement. Especially, we introduce a mathematical and logical approach to derive entanglement from POVM-based coherence. This is interesting from the perspective of quantum resource theory. We expect that with physical and more operational interpretations of POVM-coherence in future, there will be interesting applications in quantum information.

A further positive aspect of our present study is the following observation. 
Quantum states in small spaces are clearly easy to maintain and transport, so inducing entanglement from small systems to larger systems will provide remarkable benefits for resource generation and maintenance. From this viewpoint, POVM-based coherence has the advantage of resource retention and can provide maximal entanglement in large-dimensional systems depending on the given POVM and original quantum states. Furthermore, this scenario allows us to optimize the amount of entanglement generated in a given resource state by selecting different POVMs. 

The paper is structured as follows. Section \uppercase\expandafter{\romannumeral2} introduces block-coherence and POVM-based coherence. In Sec. \uppercase\expandafter{\romannumeral3}, we propose a protocol that generates entanglement from block-coherence via a bipartite incoherent operation.
In Sec. \uppercase\expandafter{\romannumeral4}, we explain that POVM-based coherence associated with block coherence through the Naimark extension, acts as a potential resource from the perspective of generating entanglement. We discuss strategies that require embedding channels and auxiliary systems for creating entanglement from POVM-based coherence, first in some simple systems and later generalize them.
Section \uppercase\expandafter{\romannumeral5} concludes with a summary.

\section{Resource theory of block-coherence and POVM-based coherence}

\subsection{Block-coherence}

We know that the incoherent or coherence-free states are associated with quantum measurement (von Neumann measurement). From this point on, for a projective measurement $\bP = \{P_i\}$ in which the orthogonal operators $P_i$ have arbitrary ranks, Bischof \emph{et al.} \cite{Bischof1} presented the resource theory of block-coherence based on the general measures, for the degree of superposition introduced by {\AA}berg \cite{Aberg}, as follows:

\begin{enumerate}
  \item The set $\cI_{\textmd{BI}}$ of block incoherent (or free) states is defined by
\be
\cI_{\textmd{BI}} = \{\rho_{\textmd{BI}} | \rho_{\textmd{BI}} = \sum_iP_i \sigma P_i = \Delta[\sigma],\quad \sigma \in\cD \}
\ee
where $\cD$ is the set of quantum states and $\Delta$ is defined as the block-dephasing map.
  \item A completely-positive and trace-preserving (CPTP) map $\Lambda_{\textmd{MBI}}$ is said to be a (maximally) block-incoherent operation if and only if it maps any block incoherent state to a block-incoherent state, i.e., $\Lambda_{\textmd{MBI}}[\cI_{\textmd{BI}}] \subseteq \cI_{\textmd{BI}}$.
We denote the set of all (maximally) block-incoherent operations as $\cO_{\textmd{MBI}}$.
\end{enumerate}

A functional $C(\rho,\bP)$ on the space of quantum states is regarded as a block-coherence measure with respect to the projective measurement $\bP$ if it fulfills the following:
B1 (\emph{Faithfulness}): $C(\rho,\bP) \geq 0$, and $C(\sigma,\bP) = 0$ if and only if $\sigma\in \cI_{\textmd{BI}}$.
B2 (\emph{Monotonicity}): $C(\Lambda_{\textmd{MBI}}[\rho],\bP)\leq C(\rho,\bP)$ for all $\Lambda_{\textmd{MBI}}\in \cO_{\textmd{MBI}}$.
B3 (\emph{Convexity}): $C(\rho,\bP)$ is convex in $\rho$.

Several measures of block-coherence have been proposed in \cite{Aberg, Bischof1}. Here we introduce two block-coherence measures for our purpose: distance-based block-coherence $C_d$ and the relative entropy of block-coherence $C_r$, defined as
\be
C_d(\rho,\bP) =\min_{\sigma\in\cI_{\textmd{BI}}}d(\rho,\sigma),
\ee
where the distance $d$ is contractive, $i.e.,$ $d(\Lambda[\rho],\Lambda[\sigma]) \leq d(\rho,\sigma)$ for any CPTP map $\Lambda$, and
\be\label{eq:3}
C_r(\rho,\bP) =\min_{\sigma\in\cI_{\textmd{BI}}}S(\rho\parallel\sigma) = S(\Delta[\rho]) - S(\rho),
\ee
where  $S(\rho) = -\tr{\rho\log\rho}$ is the von Neumann entropy and $S(\rho\parallel\sigma) = \tr{\rho\log\rho - \rho\log\sigma}$ is the quantum relative entropy.

These measures of block-coherence quantify the degree of the superposition of quantum states across different subspaces given by the range of each $P_i$.

\subsection{POVM-based coherence}

Bischof \emph{et al.} \cite{Bischof1} generalized the resource theory of coherence defined with respect to arbitrary positive-operator-valued measures, not limited to projective measurements.
The generalisation of the coherence requires the following theorem:
\begin{itemize}
  \item (\emph{Naimark Theorem}) Let $\bE = \{E_i\}_{i=0}^{n-1}$ be a POVM on $\cH^{S_0}$ with $n$ outcomes. Then, the POVM $\bE$ can be extended to a projective measurement $\bP = \{P_i\}_{i=0}^{n-1}$ on the Naimark space $\cH^{S}$ of dimension $d_{S}\geq d_{S_0}$ , such that
\bea
\tr{E_i\rho} = \textmd{tr}\big\{P_i(\rho\ot\out{0}{0})\big\}
\eea
holds for all states $\rho$ in $\cS_0$.
\end{itemize}
In the generalized resource theory of coherence, they proposed the following two definitions using Naimark theorem:
\begin{enumerate}
        \item (\emph{POVM-based coherence measure}) The POVM-based coherence measure $C(\rho^{S_0},\bE)$ for a state $\rho^{S_0}$ in $\cS_0$ is defined from the block coherence of the embedded state $\Phi[\rho^{S_0}] = \rho^{S_0}\ot \out{0}{0}^{S_1}$ with respect to a canonical Naimark extension $\bP$ of the POVM $\bE$ as follows,
\be\label{eq:4}
C(\rho^{S_0},\bE) := C(\Phi[\rho^{S_0}],\bP),
\ee
where $C(\rho^{S},\bP)$ is a unitarily invariant block-coherence measure on $\cS$, the set of quantum states on the Naimark space $\cH^S$.

     \item (\emph{POVM-incoherent operations}) A (maximally) POVM-incoherent operation is defined as $\Lambda_{\textmd{MPI}} = \Phi^{-1}\circ \Lambda\circ \Phi$, where $\Lambda$ is a block incoherent operation  with respect to $\bP$ that satisfies
\be\label{eq:5}
\Lambda[\cS_{\Phi}] \subseteq \cS_{\Phi}
\ee
 for the subset $\cS_{\Phi} \subseteq \cS$ of embedded system states $\Phi[\rho^{S_0}]$.
\end{enumerate}
This POVM-based coherence is well defined having the following properties \cite{Bischof1} :
P1 (\emph{Faithfulness}) and P3 (\emph{Convexity}): The nonnegativity and convexity of $C(\rho^{S_0},\bE)$ in $\cS_0$ follow directly from the underlying block-coherence measure. $C(\sigma^{S_0},\bE) = 0$ if and only if
\be\label{eq:6}
\sum_i\overline{E}_i\sigma^{S_0}\overline{E}_i = \sigma^{S_0},
\ee
where $\overline{E}_i$ denotes the projective part of $E_i$, \emph{i.e.}, the projector onto the range of $E_i$.
P2 (\emph{Monotonicity}): $C(\Lambda_{\textmd{MPI}}[\rho],\bE) \leq C(\rho,\bE)$ for all POVM-incoherent operations $\Lambda_{\textmd{MPI}}$ of the POVM $\bE$.

Then, the relative entropy of POVM-based coherence measure $C_r(\rho^{S_0},\bE)$, which is convex and independent of the choice of Naimark extension, is derived from the relative entropy of block-coherence measure in Eq.(\ref{eq:4}). It is expressible in the form:
\be
C_r(\rho^{S_0},\bE) = H[\{p_i\}] + \sum_ip_iS(\rho_i) - S(\rho^{S^0}),
\ee
where $p_i = \tr{E_i\rho^{S_0}},\ \rho_i$ is the $i$-th postmeasurement state for a given measurement operator $A_i$, \emph{i.e.}, $\rho_i = (1/p_i)A_i\rho^{S_0}A_i^{\dagger}$, and the Shannon entropy $H[\{p_i\}] = -\sum_ip_i\log p_i$.

In addition, Bischof \emph{et al.} showed in Ref. \cite{Bischof1} that POVM-based coherence measure and POVM-incoherent operations are independent of the choice of Naimark extension; especially the set of POVM-incoherent operations can be characterized by a semidefinite feasibility problem (SDP).
However, it is still uncertain to characterize POVM-incoherent operations without associating with Naimark extension, and impossible to induce POVM-incoherent operations associated with all block incoherent operations, \emph{i.e.}, POVM-incoherent operations are feasible only for block incoherent operations that satisfy Eq.(\ref{eq:5}). Since the task of interconversion is realized via POVM-incoherent operations, these uncertainties certainly limit direct resource conversion.
In this paper, we give a protocol for the conversion of  block coherence into entanglement via a bipartite block incoherent operation.

\section{Entanglement from block coherence}

In this section we consider the transformation of block coherence into entanglement like in the standard coherence theory in Ref. \cite{Streltsov4}.
The set $\cI_{BI}$ of bipartite block incoherent states $\rho_{BI}$ arises via a projective measurement $\bP = \{P_i\ot\out{j}{j}\}_{i,j}$ on the composite system of control system $S$ and target system $A$, i.e.,
\bea
\rho_{BI} = \sum_{i,j} (P_i\ot\out{j}{j})\sigma (P_i\ot\out{j}{j})= \Delta[\sigma],\quad \sigma\in \cD,
\eea
where $\cD$ is the set of all quantum states of the composite system $SA$.
Unlike the standard coherence, the superposition of quantum states in the subspace given by the range of each $P_i$ is not considered, but only overlaps between subspaces for different $P_i$ in quantum state are considered as resources. Therefore, the coherent state for the standard basis may not be the block coherent state.
Generally, quantum states that are simultaneously involved in projective parts of different operators in a measurement within a single system are considered to be the coherent states associated with the measurement.
We can view those measurement operators as filters, and coherence can be considered as the degree of the filter not filtered through overlay.
If we regard the measurement operators of higher ranks as a more sparse filter, this means that block coherence can be considered as a resource filtered through a more sparse quantum filter than the standard coherence.
The same is applicable in scenarios for generating entanglement from block coherence.
To realize this, we present an operation
\be\label{eq:7}
\Lambda'_{\textmd{MBI}} [\rho^{SA}]= U\rho^{SA}U^\dagger,
\ee
where the unitary operator $U$ is given by\\
\parbox{8.2cm}{
\begin{eqnarray*}
U &=& \sum_{i=0}^{d_P-1}\sum_{j=0}^{d_P-1}P_i\ot \out{((i+j)\ \textmd{mod}\  d_P)}{j}^A\\
&&+ \sum_{i=0}^{d_P-1}\sum_{j=d_P}^{d_A-1}P_i\ot \out{j}{j}^A,
\end{eqnarray*}}\hfill
\parbox{.3cm}{\begin{eqnarray}\label{eq:incoherent unitary}\end{eqnarray}} where $d_P$ is the number of $\bP = \{P_i\}$.

If $\rho^{SA}$ is a block incoherent state, then
\bea
\big[\small{P_i\ot\out{((i+k)\ \textmd{mod}\ d_P)}{k}}\big]\rho^{SA}\big[\small{P_j\ot\out{l}{((j+l)\ \textmd{mod}\ d_P)}}\big]\small{= 0}
\eea
for $i\neq j$ or $k\neq l$. This means the operation $\Lambda'_{\textmd{MBI}}$ is a block incoherent operation.
Let us consider the operation $\Lambda'_{\textmd{MBI}}$ on the states of the control system $S$.
Suppose the initial state $\rho^S$ of the control system $S$ is block coherent with respect to $\{P_i\}$. Then the operation $\Lambda'_{\textmd{MBI}}$ acting on the state $\rho^S\ot\out{0}{0}^A$ of the system $SA$ intertwine the part, caught simultaneously through different measurement operators $P_i$, with the target system $A$. This action, thus, induces entanglement between the control and the target systems.
Also, if all $P_i$ are rank one, it is an incoherent operation that creates entanglement from the standard coherence (cf. Eq. (9) in Ref. \cite{Streltsov4}). We arrive at the following theorem with respect to converting block coherence into entanglement:

\begin{thrm}\label{thm 2}
A state $\rho^S$ can be converted to an entangled state via block incoherent operations if and only if $\rho^S$ is block coherent.
\end{thrm}
\begin{proof}
The proof of this theorem is similar to that of Theorem 2 in Ref. \cite{Streltsov4}. We first prove that for any contractive distance $d$, the amount of entanglement $E_d$ obtainable from a state $\rho^S$ via a block incoherent operation $\Lambda_{\textmd{MBI}}$ is bounded above by the block-coherence $C_d(\rho^S)$, i.e.,
\bea
E_d^{S:A}\big(\Lambda_{\textmd{MBI}}[\rho^S\ot\out{0}{0}^A]\big) \leq C_d(\rho^S,\{P_i\}),
\eea
where $E_d^{S:A}$ is the distance-based measure of entanglement defined as
\bea E_d^{S:A}(\rho^{SA}) = \min_{\sigma\in \cS_{\textmd{SEP}}}d(\rho^{SA},\sigma)
\eea
with $\cS_{\textmd{SEP}}$ being the set of all separable states.
This means that if $\rho^S$ is block incoherent then it cannot be converted into entanglement via block incoherent operations.
We consider that the distance $d$ is contractive under quantum operations,
\bea
d(\Lambda[\rho], \Lambda[\sigma]) \leq d(\rho, \sigma)
\eea
for any CPTP map $\Lambda$.
If $\sigma^S$ is the closest block incoherent state to $\rho^S$, i.e., $C_d(\rho^S,\{P_i\}) = d(\rho^S, \sigma^S)$, then we have
\bea
C_d(\rho^S,\{P_i\}) &=& d\big(\rho^S\ot\out{0}{0}^A, \sigma^S\ot\out{0}{0}^A\big)\\
&\geq& d\big(\Lambda_{\textmd{MBI}}[\rho^S\ot\out{0}{0}^A], \Lambda_{\textmd{MBI}}[\sigma^S\ot\out{0}{0}^A]\big)\\
&\geq& E_d^{S:A}\big(\Lambda_{\textmd{MBI}}[\rho^S\ot\out{0}{0}^A]\big),
\eea
where equality and the first inequality follow from the contractivity of the distance $d$, and the second inequality is due to the definition of $E_d^{S:A}$ and $\cI_{\textmd{BI}} \subset \cS_{\textmd{SEP}}$.

Conversely, we consider that the block incoherent operation $\Lambda'_{\textmd{MBI}}$ in Eq.(\ref{eq:7}) maps the initial state $\rho^S\ot\out{0}{0}^A$ to the state
\be\label{eq:incoherent operation}
\Lambda'_{\textmd{MBI}}[\rho^S\ot\out{0}{0}^A] = \sum_{i,j}P_i\rho^S P_j\ot\out{i}{j}^A.
\ee
Note that the matrix elements of $\rho^S$ are embedded in the matrix of $\Lambda'_{\textmd{MBI}}[\rho^S\ot\out{0}{0}^A]$ as follows (the other matrix elements of $\Lambda'_{\textmd{MBI}}[\rho^S\ot\out{0}{0}^A]$ are all zero):
\bea
\rho_{k^{(i)}l^{(j)}} &=&\bra{k^{(i)}}\rho^S\ket{l^{(j)}} = \bra{k^{(i)}}P_i\rho^S P_j\ket{l^{(j)}}\\
&=& \bra{k^{(i)}i}\Lambda'_{\textmd{MBI}}[\rho^S\ot\out{0}{0}^A]\ket{l^{(j)}j},
\eea
where $\{\ket{k^{(i)}}\}_k$ is a basis of subspace given by the range of $P_i$ and $\ket{mn} = \ket{m}\ot\ket{n}$. It implies that
\bea
S(\rho^S) = S\big(\Lambda'_{\textmd{MBI}}[\rho^S\ot\out{0}{0}^A]\big).
\eea
Then, for the relative entropy of entanglement $E_r^{S:A}$ defined as
\bea
E_r^{S:A}(\rho^{SA}) = \min_{\sigma\in\cS_{\textmd{SEP}}}S(\rho^{SA} \parallel \sigma),
\eea
we have\\
\parbox{8.2cm}{
\begin{eqnarray*}
E_r^{S:A}(\Lambda'_{\textmd{MBI}}[\rho^S\ot\out{0}{0}^A]) &\geq& S(\Delta[\rho^S])-S(\rho^S)\\
&=& C_r(\rho^S,\{P_i\}),
\end{eqnarray*}}\hfill
\parbox{.3cm}{\begin{eqnarray}\label{eq:relative entropy of entanglement}\end{eqnarray}}\\
where inequality is due to $E_r^{S:A}(\sigma^{SA})\geq S(\sigma^{S}) - S(\sigma^{SA})$ \cite{Plenio} and $\ptr{A}{\Lambda'_{\textmd{MBI}}[\rho^S\ot\out{0}{0}^A]} = \Delta[\rho^S]$ , and equality follows from Eq.(\ref{eq:3}).
Therefore, if $C_r(\rho^S, \{P_i\})>0$, \emph{i.e.}, $\rho^S$ is block coherent, the state $\rho^S$ can be converted to an entangled state via a block incoherent operation.

\end{proof}

In particular, there is always a block incoherent operation $\Lambda'_{\textmd{MBI}}$ that satisfies
\be\label{eq:12}
E_r^{S:A}\big(\Lambda'_{\textmd{MBI}}[\rho^S\ot\out{0}{0}^A]\big) = C_r^B(\rho^S,\{P_i\}),
\ee
using the contractivity of the relative entropy for the unitary operation \cite{Watrous} and Eq. (\ref{eq:relative entropy of entanglement}).
This means that there must be a block incoherent operation that converts the same degree of entanglement as the block coherence in the initial state of $S$ and that we can control the density of filters that measure coherence by choosing each $P_i$ from different ranks, and it is possible to control the degree of entanglement converted using this density.

\section{Entanglement from POVM-based coherence}

In this section we consider the transformation of POVM-based coherence into entanglement. 
We remind the readers that there are some potential constraints here.
The first is that the existence of POVM-incoherent operations cannot be guaranteed without considering Naimark extension and that, even if considered, not all block incoherent operations induce the associated POVM-incoherent operations. Especially, the block incoherent operation of Eq.(\ref{eq:incoherent operation}) that causes all coherence to be converted into entanglement does not satisfy the condition in Eq.(\ref{eq:5}), so the POVM-incoherent operation associated with Eq.(\ref{eq:incoherent operation}) is not realized. And the second is that entanglement of $S:A$ produced by block coherence considered in the Naimark extension does not guarantee the generation of entanglement of $S_0:A$ in the original system; all entanglement can exist only as the entanglement of $S_1:A$ (cf. Example \ref{ex:3}).
This means the generation of entanglement through POVM-incoherent operators from POVM-based coherence in the initial state cannot be considered in all situations and that entanglement of the original system may not be extracted from the generated entanglement on the Naimark system.

Following, we propose some strategies to generate entanglement from POVM-based coherence according to the methods of Naimark extension. The Naimark extension can be realized in various ways namely using the embedding channel and attaching an auxiliary system to the original system. We present some examples to illustrate the same. Example \ref{ex:3} shows that all entanglement in a quantum state can exist only between $S_1$ and $A$. In Example \ref{ex:3a} we illustrate the generation of maximal entanglement using embedding channel, and  we consider harnessing entanglement through Naimark extension which requires auxiliary systems in Example \ref{ex:4}.\\

\begin{exam}\label{ex:3}
Let $\bE = \{E_i\}_{i=0}^1$ be a POVM of a qubit system $S_0$, where
\bea
E_0 &=& a\out{0}{0}+(1-a)\out{1}{1},\\
E_1 &=& (1-a)\out{0}{0} + a\out{1}{1}\quad (0< a<\frac{1}{2}).
\eea
We also construct a Naimark extension $\bP = \{P_i\}_{i=0}^1$ of $\bE$ as
\bea
P_i = \sum_{\scriptsize\begin{array}{c}
              a=0,1 \\
              b=0,1
            \end{array}
}A^\dagger_{i,a}A_{i,b}\ot \out{a}{b}
\eea
where
\bea
A_{i,0} &=& A_{i}= \sqrt{E_i} \ (i=0,1), \\
A_{1,1} &=& -\sqrt{1-a}\out{0}{0} + \sqrt{a}\out{1}{1}, \\
A_{2,1} &=& \sqrt{a}\out{0}{0} - \sqrt{1-a}\out{1}{1}.
\eea
Then the block incoherent unitary operation of (\ref{eq:incoherent operation}) maps the state $\rho^{S_0}\ot \out{1}{1}^{S_1}\ot\out{1}{1}^A$ to the state
\bea
&&\Lambda'_{\textmd{MBI}}[\rho^{S_0}\ot \out{0}{0}^{S_1}\ot\out{0}{0}^A]\\
&& = \sum_{i.j}\sum_{a,b}A^\dagger_{i,a}A_{i}\rho^{S_0}A^\dagger_{j}A_{j,b}\ot\out{a}{b}^{S_1}\ot\out{i}{j}^{A}.
\eea

In addition, $\out{0}{0}$ and $\out{1}{1}$ as the eigen states of $E_0$ and $E_1$ on $S_0$ are POVM-based coherent states for $\bE$, and by the above operation the coherence is converted to nonzero-entanglement between $S$ and $A$.
But if the initial state $\rho^{S_0}$ is $\out{k}{k}~(k=0,1)$, the quantum state after the unitary transformation is
\bea
&&\Lambda'_{\textmd{MBI}}[\rho^{S_0}\ot \out{0}{0}^{S_1}\ot\out{0}{0}^A]\\
&& = \out{k}{k} \ot \sum_{i.j}\sum_{a,b}\beta^{(k)}_{i,a}\beta^{(k)}_{j,b}\out{a}{b}^{S_1}\ot\alpha^{(k)}_{i,j}\out{i}{j}^{A},
\eea
where $\alpha^{(k)}_{i,j} = \bra{k}A_{i}\out{k}{k}A^\dagger_{j}\ket{k}, \ \beta^{(k)}_{i,a} = \bra{k}A_{i,a}\ket{k}$\ $(k=0,1)$.
That is, all entanglement in this quantum state exists only between $S_1$ and $A$.
\end{exam}

Thus, it is possible to create entanglement in the Naimark system from POVM-based coherence through the extension. Note that embedding the original system into the expanded system is required to realize the Naimark extension. This is usually realized through the embedding channel or by attaching the auxiliary state to the initial state.

\subsection{Maximal entanglement using embedding channel}

We first introduce Naimark extension using the embedding channel and the strategy of generating entanglement through it. Through the strategy we can generate the maximal entanglement in the Naimark space, even if the initial state was prepared in a space smaller than the Naimark space. Here we illustrate it by giving an example.

\begin{exam}\label{ex:3a}
\textbf{(Three elements (qubit trine) POVM)}
Let $\bE^{(3)} = \{\frac{2}{3}\out{\phi_k}{\phi_k}\}_{k=0}^2$ be a POVM on $\cH^{S_0} = \cC^2$ with $\ket{\phi_k} = 1/\sqrt{2}(\ket{0}+\omega^{k}\ket{1})$, where $\omega = \exp(2\pi i/3)$, and $\bP^{(3)} = \{\out{\varphi_k}{\varphi_k}\}_{k=0}^2$ be a Naimark extension of it on $\cH^{S} = \cC^3$, where
\bea
\ket{\varphi_0}&=&\frac{1}{\sqrt{3}}\big(\ket{0}+\ket{1}+\ket{2}\big)\\
\ket{\varphi_1}&=&\frac{1}{\sqrt{3}}\big(\ket{0}+\omega\ket{1}+\omega^2\ket{2}\big)\\
\ket{\varphi_2}&=&\frac{1}{\sqrt{3}}\big(\ket{0}+\omega^2\ket{1}+\omega\ket{2}\big).
\eea
From the expression of $\bP^{(3)}$, we can see that positive measurement is a rank-one projective measurement, i.e., von Neumann measurement.
Therefore, we can generate the entanglement on $S:A$ through the incoherent operation $\Lambda_{SA}[\rho^{S}\ot\out{0}{0}^A] = U(\rho^{S}\ot\out{0}{0}^A)U^\dagger$ for all coherent states $\rho^{S}$ on the system $S \ (d_A\geq3)$ where
\bea
U &=& \sum_{i=0}^{2}\sum_{j=0}^{2}\out{\varphi_i}{\varphi_i}\ot \out{((i+j)\ \textmd{mod}\ 3)}{j}^A\\
&&+ \sum_{i=0}^{2}\sum_{j=3}^{d_A-1}\out{\varphi_i}{\varphi_i}\ot \out{j}{j}^A.
\eea

In addition, all states $\rho^{S_0}$ on the qubit system have nonzero POVM-based coherence because there is no quantum state satisfying Eq. (\ref{eq:6}). Therefore, we have the following result for relative entropy of entanglement from Eqs. (\ref{eq:4}) and (\ref{eq:12}) :
\bea
E_r^{S:A}\big(\Lambda_{SA}[(\rho^{S_0}\oplus\ 0) \ot\out{0}{0}^A]\big) = C_r(\rho^{S_0},\bE^{(3)}) > 0
\eea
with $\rho^{S_0}\oplus 0 = \left(
                          \begin{array}{ccc}
                            \rho^{00} & \rho^{01} & 0 \\
                            \rho^{10} & \rho^{11} & 0 \\
                            0 & 0 & 0 \\
                          \end{array}
                        \right)
$ being a general embedded system state on $\cH^{S} = \cC^3$ where $\rho^{ij} = \bra{i}\rho^{S_0}\ket{j}$. This means that all quantum states on $\cC^{2}$ can be utilized as potential resources on $\cC^3$.
In particular, when the state $\rho^{S_0}$ is $\out{0}{0}$ or $\out{1}{1}$ with maximal POVM-based coherence of $\log3$, we have
\bea
\Lambda_{SA}\big[\out{0}{0}^{S} \ot\out{0}{0}^A\big] &=& \sum_{i,j=0}^2\frac{1}{3}\out{\varphi_i}{\varphi_j}\ot \out{i}{j}^A, \\
\Lambda_{SA}\big[\out{1}{1}^{S} \ot\out{0}{0}^A\big] &=& \sum_{i,j=0}^2\frac{1}{3}\out{\varphi'_i}{\varphi'_j}\ot \out{i}{j}^A,
\eea
where $\out{i}{i}^{S}$ be a general embedded system state on $\cC^3$ of $\out{i}{i} \ (i=0,1)$, and
\bea
\ket{\varphi'_0}&=&\ket{\varphi_0}\\
\ket{\varphi'_1}&=&\frac{1}{\sqrt{3}}\big(\omega^2\ket{0}+\ket{1}+\omega\ket{2}\big)\\
\ket{\varphi'_2}&=&\frac{1}{\sqrt{3}}\big(\omega\ket{0}+\ket{1}+\omega^2\ket{2}\big).
\eea
This gives $E_r^{S:A}\big(\Lambda_{SA}[(\rho^{S_0}\oplus\ 0) \ot\out{0}{0}^A]\big) = \log3$. Thus, after embedding the quantum state on the qubit system, the maximum entanglement on the larger system $\cC^3 \ot \cH^A$ can be obtained through an incoherent operation.
\end{exam}

Furthermore, we can generalize the above result to a larger Naimark extension system provided the dimension of the system is prime.
When $d_{S}(\geq3)$ is a prime number, we consider the POVM $\bE = \{\frac{2}{d}\out{\phi_k}{\phi_k}\}_{k=0}^{d_{S}-1}$ on $\cH^{S_0} = \cC^2$ with $\ket{\phi_k} = 1/\sqrt{2}(\ket{0}+\nu^{k}\ket{1})$, where $\nu = \exp(2\pi i/d_{S})$, and a Naimark extension $\bP$ of $\bE$ on $\cH^{S} = \cC^{d_{S}}$ is given by
$\bP = \{\out{\varphi_k}{\varphi_k}\}_{k=0}^{d_A-1}$, where
\bea
\ket{\varphi_k}=\frac{1}{\sqrt{d_{S}}}\Big(\sum_{i=0}^{d_{S}-1}\nu^{(ik\ \textmd{mod}\ d_{S})}\ket{i}\Big)
\eea
for $k = 0,1,\cdots,d_{S}-1.$

Then, for all quantum states $\rho^{S_0} = \left(
                          \begin{array}{cc}
                            \rho^{00} & \rho^{01}\\
                            \rho^{10} & \rho^{11}\\
                          \end{array}
                        \right)$ on the system $S_0$, there is a bipartite incoherent operation $\Lambda_{SA}$ on $\cH^{SA}$ that can generate nonzero entanglement from the general embedded system state
\bea
\rho^{S_0}+0 = \left(
                          \begin{array}{ccccc}
                            \rho^{00} & \rho^{01} &  0 & \cdots & 0\\
                            \rho^{10} & \rho^{11} &  0 & &\\
                            0 & 0 &  0 &  & \vdots \\
                            \vdots &  &  & \ddots &  \\
                            0 &  & \cdots &  & 0 \\
                          \end{array}
                        \right)
                        \eea
on $\cH^{S}$, that is $\Lambda_{SA}[(\rho^{S_0}+0)\ot\out{0}{0}^A] = U[(\rho^{S_0}+0)\ot\out{0}{0}^A]U^\dagger$ $(d_A \geq d_{S})$ where
\bea
U &=& \sum_{i=0}^{d_{S}-1}\sum_{j=0}^{d_{S}-1}\out{\varphi_i}{\varphi_i}\ot \out{((i+j)\ \textmd{mod}\ d_{S})}{j}^A\\
&&+ \sum_{i=0}^{d_{S}-1}\sum_{j=d_{S}}^{d_A-1}\out{\varphi_i}{\varphi_i}\ot \out{j}{j}^A.
\eea

\begin{figure}[t]
\centering
\includegraphics[height=.33\textheight]{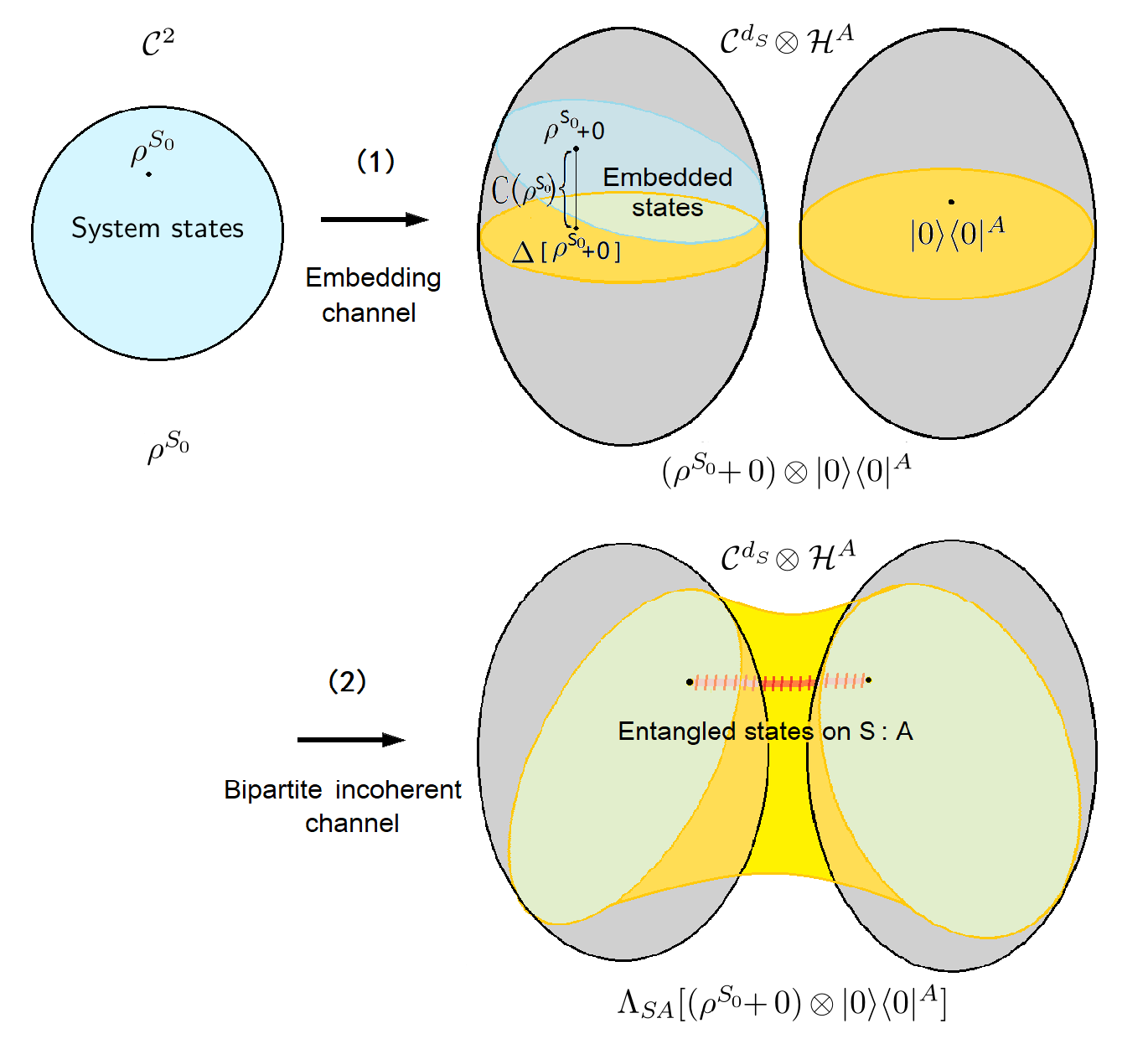}
\caption{\label{fig:1} (1) The qubit states $\rho^{S_0}$ with POVM-based coherence as a potential resource are embedded by the embedding channel into quantum states with coherence associated with a Naimark extension $\bP$ of the POVM $\bE$, as a practical resource.  \ (2) POVM-based coherence in the embedded states is converted to entanglement on $S:A$ by a bipartite incoherent channel for $\bP$. If the dimension $d_{S}$ of the
Naimark space is prime, the maximum entanglement $\log d_{S}$ can be generated from the maximal POVM-based coherent states $\out{0}{0}$ and $\out{1}{1}$.}
\end{figure}

This framework, which translates more entanglement from the qubit system $\cC^2$, requires two processes: the embedding channel (Fig. \ref{fig:1}-(1)) and the bipartite incoherent channel (Fig. \ref{fig:1}-(2)).
Let us consider these embedding and bipartite incoherent channels for $\out{0}{0}$ and $\out{1}{1}$  with maximal POVM-based coherence of $\log d_{S}$ on $\cH^{S_0}$.  First, the state $\out{0}{0}$ (or $\out{1}{1}$) is embedded to $\out{0}{0}^{S}\ot\out{0}{0}^A$ (or $\out{1}{1}^{S}\ot\out{0}{0}^A$) via the embedding channel, where $\out{i}{i}^{S}$ be a general embedded system state on $\cC^{d_{S}}$ of $\out{i}{i} \ (i=0,1)$. Next, we convert the embedded state into the following state with maximal entanglement of $\log d_{S}$ by the bipartite incoherent channel :
\bea
\Lambda_{SA}\big[\out{0}{0}^{S} \ot\out{0}{0}^A\big] &=& \sum_{i,j=0}^{d_{S}-1}\frac{1}{d_{S}}\out{\varphi_i}{\varphi_j}\ot \out{i}{j}^A, \\
\Lambda_{SA}\big[\out{1}{1}^{S} \ot\out{0}{0}^A\big] &=& \sum_{i,j=0}^{d_{S}-1}\frac{1}{d_{S}}\out{\varphi'_i}{\varphi'_j}\ot \out{i}{j}^A,
\eea
where
$$\ket{\varphi'_k}=\frac{1}{\sqrt{d_{S}}}\big(\nu^{((d_{S}-1)k\ \textmd{mod}\ d_{S})}\ket{0} + \sum_{i=0}^{d_{S}-2}\nu^{(ik\ \textmd{mod}\ d_{S})}\ket{i}\big).$$
This allows us to create the maximal entanglement of $\log d_{S}$ on $S:A$ via the incoherent operation $\Lambda_{SA}$ when the quantum state $\rho^{S_0}$ is $\out{0}{0}$ or $\out{1}{1}$.

From the perspective of the system to handle the quantum states of each $\cC^2$ and $\cC^{d_{S}}$, obviously the quantum states on the qubit system $\cC^2$ are more advantageous for transmission and storage than those on the system with larger dimensions. We can also generate entanglement, even the maximal entanglement, on the target system from initial states on $\cC^2$ by preparing corresponding embedding and bipartite incoherent channels.

\subsection{Selective entanglement using auxiliary systems}

We consider harnessing entanglement through Naimark extension which requires auxiliary systems. The following is an example of creating entanglement from POVM (with four elements)-based coherence in a qubit system. We notice an advantage in this example: by applying bipartite block-incoherent operation only once, we can selectively implement the transformation of coherence into entanglement based on two or more POVMs.

\begin{exam}\label{ex:4}
\textbf{(Four elements POVM)}
Let $\bE^{(4)} = \{\frac{1}{2}\out{\phi_k}{\phi_k}\}_{k=0}^3$ be a POVM on $\cH^{S_0} = \cC^2$ with $\ket{\phi_k} = 1/\sqrt{2}(\ket{0}+\omega^{k}\ket{1})$, where $\omega = \exp(\pi i/2)$, and $\bP^{(4)} = \{\out{\varphi_k}{\varphi_k}\}_{k=0}^3$ be a Naimark extension of it on $\cH^{S}(=\cH^{S_0}\ot\cH^{S_1})=\cC^4$ (here $\cC^4$ can be considered the same system as $\cC^2\ot\cC^2$ by identifying $\ket{0}^{\cC^4} = \ket{00}^{\cC^2\ot\cC^2}, \ket{1}^{\cC^4} = \ket{01}^{\cC^2\ot\cC^2}, \ket{2}^{\cC^4} = \ket{10}^{\cC^2\ot\cC^2}$ and $\ket{3}^{\cC^4} = \ket{11}^{\cC^2\ot\cC^2}$), where
\bea
\ket{\varphi_0}&=&\frac{1}{2}\big(\ket{0}+\ket{1}+\sqrt{2}\ket{2}\big)\\
\ket{\varphi_1}&=&\frac{1}{2}\big(\ket{0}+i\ket{1}-\exp(\frac{\pi i}{4})\ket{2}+\exp(\frac{\pi i}{4})\ket{3}\big)\\
\ket{\varphi_2}&=&\frac{1}{2}\big(\ket{0}-\ket{1}-\sqrt{2}i\ket{3}\big)\\
\ket{\varphi_3}&=&\frac{1}{2}\big(\ket{0}-i\ket{1}-\exp(-\frac{\pi i}{4})\ket{2}+\exp(\frac{3\pi i}{4})\ket{3}\big).
\eea
Then, we can generate the entanglement on $S:A$ via the bipartite incoherent operation $\Lambda_{SA}[\rho^{S}\ot\out{0}{0}^A] = U(\rho^{S}\ot\out{0}{0}^A)U^\dagger$ for all coherent states $\rho^{S}$ on the system $S$ $(d_A\geq4)$ where
\bea
U &=& \sum_{i=0}^{3}\sum_{j=0}^{3}\out{\varphi_i}{\varphi_i}\ot \out{((i+j)\ \textmd{mod}\ 4)}{j}^A\\
&&+ \sum_{i=0}^{3}\sum_{j=4}^{d_A-1}\out{\varphi_i}{\varphi_i}\ot \out{j}{j}^A.
\eea
Also, we have the following result from POVM-based coherence theory for relative entropy of entanglement:
\bea
E_r^{S:A}\big(\Lambda_{SA}[(\rho^{S_0}\ot\out{0}{0}) \ot\out{0}{0}^A]\big) = C_r(\rho^{S_0}.\bE^{(4)}) > 0,
\eea
This is because for all quantum states on $S_0$, there are necessarily $\phi_i,\phi_j$ that satisfy the following:
\bea
\bra{\phi_i}\rho^{S_0}\ket{\phi_j} \neq 0 \quad \forall i\neq j.
\eea
It means that all quantum states on $\cC^{2}$ can become practical resources on the Naimark space $\cC^4$ by attaching auxiliary systems.
Particularly, when the state $\rho^{S_0}$ is $\out{0}{0}$ or $\out{1}{1}$ with maximal POVM-based coherence of $\log4$, we have
\bea
\Lambda_{SA}\big[\out{0}{0}\ot\out{0}{0} \ot\out{0}{0}^A\big] &=& \sum_{i,j=0}^3\frac{1}{4}\out{\varphi_i}{\varphi_j}\ot \out{i}{j}^A, \\
\Lambda_{SA}\big[\out{1}{1}\ot\out{0}{0} \ot\out{0}{0}^A\big] &=& \sum_{i,j=0}^3\frac{1}{4}\out{\varphi'_i}{\varphi'_j}\ot \out{i}{j}^A,
\eea
where
\bea
\ket{\varphi'_0}&=&\ket{\varphi_0}\\
\ket{\varphi'_1}&=&\frac{1}{2}\big(-i\ket{0}+\ket{1}-\exp(-\frac{\pi i}{4})\ket{2}+\exp(-\frac{\pi i}{4})\ket{3}\big)\\
\ket{\varphi'_2}&=&\frac{1}{2}\big(-\ket{0}+\ket{1}+\sqrt{2}i\ket{3}\big)\\
\ket{\varphi'_3}&=&\frac{1}{2}\big(i\ket{0}+\ket{1}-\exp(-\frac{\pi i}{4})\ket{2}-\exp(\frac{\pi i}{4})\ket{3}\big).
\eea
This gives $E_r^{S:A}\big(\Lambda_{SA}[(\rho^S\ot\out{0}{0}) \ot\out{0}{0}^A]\big) = \log4 = 2$ as maximal entanglement on $S:A$.
\end{exam}

In Example \ref{ex:4}, the first auxiliary system is the part added to the Naimark extension and the second system is the counterpart that shares the generated entanglement. We note an interesting thing in this example that POVM $\bE^{(4)} \equiv \bE_0$ is related only to $\out{0}{0}$ of the first auxiliary system $\cC^2$, and $\out{1}{1}$ is associated with another POVM $\bE_1$ selected under the condition that $\bP^{(4)}$ is a projection.
The elements of $\bP^{(4)}$ on $S = \cC^2 \ot \cC^2$ are expressed in the following $4\times4$ matrix forms:
\bea
\out{\varphi_0}{\varphi_0} &=& \frac{1}{4}\left(
                          \begin{array}{cc|cc}
                            1 & 1 & \sqrt{2} & 0 \\
                            1 & 1 & \sqrt{2} & 0 \\\hline
                            \sqrt{2} & \sqrt{2} & \bf{2} & \bf{0} \\
                            0 & 0 & \bf{0} & \bf{0} \\
                          \end{array}
                        \right)\\
\out{\varphi_1}{\varphi_1} &=& \frac{1}{4}\left(
                          \begin{array}{cc|cc}
                            1 & -i & -e^{-\frac{\pi}{4}i} & e^{-\frac{\pi}{4}i} \\
                            i & 1 & -e^{\frac{\pi}{4}i} & e^{\frac{\pi}{4}i} \\\hline
                            -e^{\frac{\pi}{4}i} & -e^{-\frac{\pi}{4}i} & \bf{1} & \bf{-1} \\
                            e^{\frac{\pi}{4}i} & e^{-\frac{\pi}{4}i} & \bf{-1} & \bf{1} \\
                          \end{array}
                        \right)\\
\out{\varphi_2}{\varphi_2} &=& \frac{1}{4}\left(
                          \begin{array}{cc|cc}
                            1 & -1 & 0 & i\sqrt{2} \\
                            -1 & 1 & 0 & -i\sqrt{2} \\\hline
                            0 & 0 & \bf{0} & \bf{0} \\
                            -i\sqrt{2} & i\sqrt{2} & \bf{0} & \bf{2} \\
                          \end{array}
                          \right)\\
\out{\varphi_3}{\varphi_3} &=& \frac{1}{4}\left(
                          \begin{array}{cc|cc}
                            1 & i & -e^{\frac{\pi}{4}i} & -e^{\frac{\pi}{4}i} \\
                            -i & 1 & -e^{-\frac{\pi}{4}i} & e^{\frac{3\pi}{4}i} \\\hline
                            -e^{-\frac{\pi}{4}i} & -e^{\frac{\pi}{4}i} & \bf{1} & \bf{1} \\
                            e^{\frac{3\pi}{4}i} & -e^{\frac{\pi}{4}i} & \bf{1} & \bf{1} \\
                          \end{array}
                        \right).
\eea
The elements of $\bE_1$ on the original system $\cC^2$ are determined by $2\times2$ matrices in the lower right (the blocks in bold font) as $\bE_1 = \{\frac{1}{2}\out{\phi^{(1)}_j}{\phi^{(1)}_j}\}_{j=0}^3$ with
\bea
\ket{\phi^{(1)}_0}&=&\ket{0}\\
\ket{\phi^{(1)}_1}&=&\frac{1}{\sqrt{2}}\big(-\exp(\frac{\pi i}{4})\ket{0}+\exp(\frac{\pi i}{4})\ket{1}\big)\\
\ket{\phi^{(1)}_2}&=&-i\ket{1}\\
\ket{\phi^{(1)}_3}&=&\frac{1}{\sqrt{2}}\big(-\exp(-\frac{\pi i}{4})\ket{0}+\exp(\frac{3\pi i}{4})\ket{1}\big).
\eea
In other words, the operators of the projection measurement $\bP^{(4)} = \{\out{\varphi_j}{\varphi_j}\}_{j=0}^3$ are composed of cross-combination of operators of $\bE^{(4)}$ and $\bE_1$ on Naimark space as follows:
\bea
\ket{\varphi_j} = \frac{1}{\sqrt{2}}\sum_{a=0}^1 \ket{\phi^{(a)}_j}\ot\ket{a},
\eea
where $\ket{\phi^{(0)}_j} = \ket{\phi_j}$ for $j=0,1,2,3$. It implies that $\bP^{(4)}$ is a Naimark extension for both $\bE^{(4)}$ and $\bE_1$ because
\bea
&&\bra{\varphi_j}(\rho^S\ot\out{a}{a})\ket{\varphi_j}\\
&& =\big(\frac{1}{\sqrt{2}}\big)^2\sum_{k,l}\bra{\phi^{(k)}_j}\rho^S\ket{\phi^{(l)}_j}\iinner{k}{a}\iinner{a}{l}\\
&& =\frac{1}{2}\bra{\phi^{(a)}_j}\rho^S\ket{\phi^{(a)}_j},
\eea
for $a=0,1$. Moreover, $\bE_1$ is not an arbitrary POVM. It is definitely a different measurement from $\bE^{(4)}$ which allows us to measure the amount of coherence based on corresponding POVM depending on the quantum state ($\ket{0}$ or $\ket{1}$) on the first auxiliary system added for Naimark extension, that is,
\bea
C(\rho^{S_0}\ot\out{a}{a},\bP^{(4)}) = C(\rho^{S_0},\bE_a),
\eea
where $(a=0,1)$ and $\bE_0 = \bE^{(4)}$.

\begin{figure}[t]
\centering
\includegraphics[height=.7\textheight]{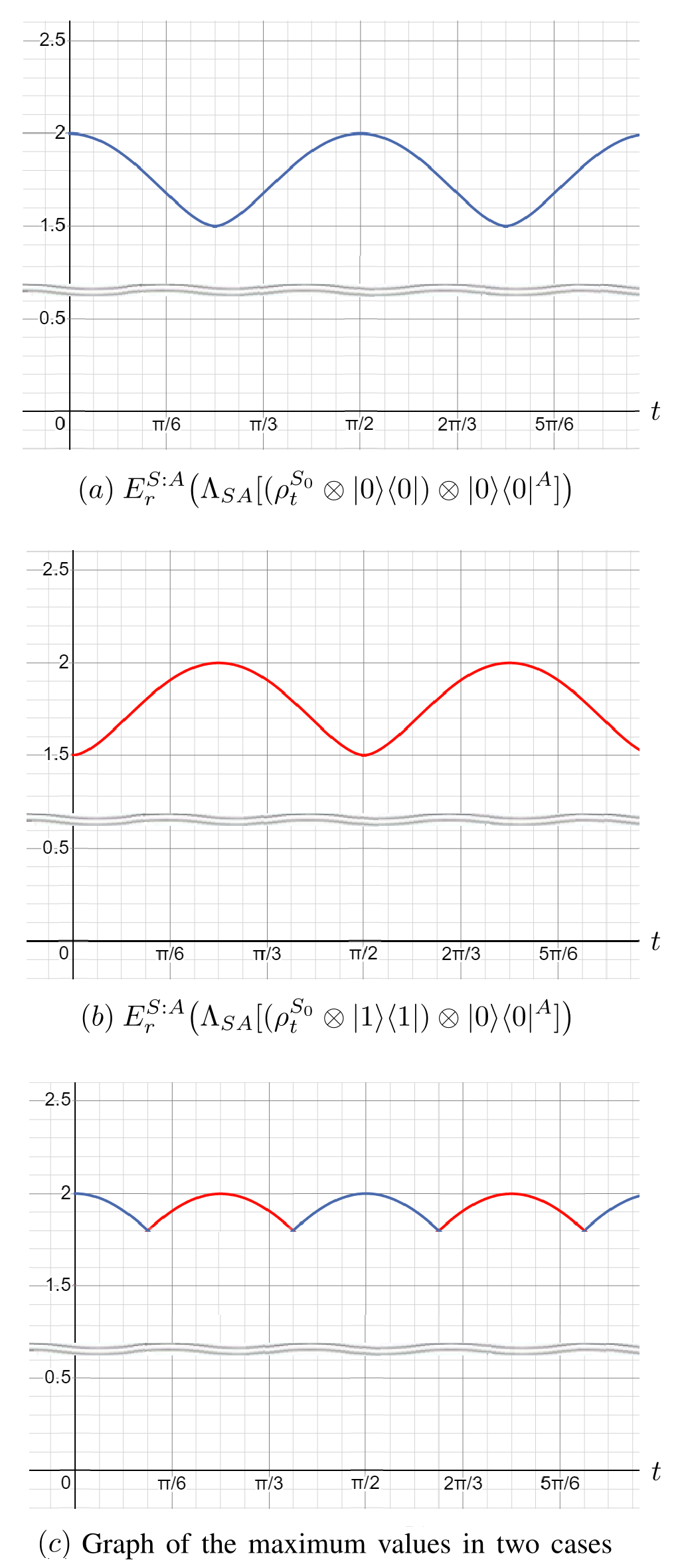}
\caption{\label{fig:fig-2} Graph of the amount of entanglement generated from $\rho^{S_0}_t$ by $\Lambda_{SA}$ when the first auxiliary state is fixed at (a) $\ket{0}$ (blue line) and (b) $\ket{1}$ (red line). (c) Graph of the amount of entanglement obtained from $\rho^{S_0}_t$ due to $\Lambda_{SA}$ by adjusting the first auxiliary state. There is maximal entanglement of $\log4$ at $t=0,\frac{\pi}{4},\frac{\pi}{2},$ and $\frac{3\pi}{4}$.}
\end{figure}

These $\bE_0$ and $\bE_1$ are measurements with the same resultant values for pure states in vertical positions above the circumference that pass through $\ket{0}$ and $\frac{1}{\sqrt{2}}(\ket{0}+i\ket{1})$ on the Bloch sphere surface. And $\Lambda_{SA}$ is a bipartite incoherent operation based on $\bP^{(4)}$ and $\{\out{i}{i}^A\}_{i=0}^{d_A-1}$, so it acts the same way in converting POVM-based coherence for $\bE_0$ and $\bE_1$ respectively into entanglement by mapping the state $\rho^{S_0}\ot\out{a}{a}\ot\out{0}{0}^A\ (a=0,1)$ to the state
\bea
&&\Lambda_{SA}[\rho^{S_0}\ot\out{a}{a}\ot\out{0}{0}^A] \\
&& = \frac{1}{2} \sum_{j,k}\bra{\phi^{(a)}_j}\rho^{S_0}\ket{\phi^{(a)}_k}\out{\varphi_j}{\varphi_k}\ot\out{j}{k}^A.
\eea
Then, for relative entropy of entanglement, it implies that
\bea
E_r^{S:A}\big(\Lambda_{SA}[(\rho^{S_0}\ot\out{a}{a}) \ot\out{0}{0}^A]\big) = C_r(\rho^{S_0},\bE_a).
\eea
Consider the phase state $\rho^{S_0}_t = \out{\psi_t}{\psi_t}$ with $\ket{\psi_t} = \cos t\ket{0}+ i\sin t\ket{1}\ (t\in[0,\pi))$. Then, when $t = 0$ or $\frac{\pi}{2}$, we have
\bea
&& E_r^{S:A}\big(\Lambda_{SA}[(\rho^{S_0}_t\ot\out{0}{0}) \ot\out{0}{0}^A]\big) = \log4 = 2\\
&&\textmd{and} \quad E_r^{S:A}\big(\Lambda_{SA}[(\rho^{S_0}_t\ot\out{1}{1}) \ot\out{0}{0}^A]\big) = 1.5.
\eea
On the contrary, when $t = \frac{\pi}{4}$ or $\frac{3\pi}{4}$, we have
\bea
 E_r^{S:A}\big(\Lambda_{SA}[(\rho^{S_0}_t\ot\out{0}{0}) \ot\out{0}{0}^A]\big) &=& 1.5\quad \textmd{and}\\
E_r^{S:A}\big(\Lambda_{SA}[(\rho^{S_0}_t\ot\out{1}{1}) \ot\out{0}{0}^A]\big) &=& \log4 = 2.
\eea

This shows the amount of entanglement generated from the same initial state can vary depending on the first auxiliary state, chosen between $\ket{0}$ and $\ket{1}$, and the initial states to generate the maximal entanglement of $\log4$ on $S:A$ are also different. Therefore, we can adopt a strategy that generates more entanglement by attaching different auxiliary states depending on the initial state (see Fig. \ref{fig:fig-2}).

This is valid for all Naimark extensions requiring auxiliary systems. The Naimark extension, which requires auxiliary systems, is generally composed of $n$ POVMs that are cross-coupled. We have just looked into the case of POVM which consists of pure state operators (this can be generalized analogously in the Naimark spaces with dimensions of $2\times n$). Next, we see that the above result remains the same even if the elements of POVM are general positive operators.
Let $\bE_0 = \{E_i^{(0)}\}_{i=0}^{n-1}$ be a POVM with $n$ general positive operators on $\cH^{S_0}$ as elements. Then, using the unitary operator $V = \sum_{i, a} A_{i,a} \ot \out{i}{a}$ (see Supplemental Material of Ref. \cite{Bischof1}) where $A_i$ is any measurement operator, the canonical Naimark extension $\bP = \{P_i\}_{i=0}^{n-1}$ on $\cH^S = \cH^{S_0}\ot\cH^{S_1} \ (d_{S_1}= n)$ of $\bE_0$ is defined as 
\bea
P_i &:=& V^{\dagger} \big(\id_{S_0} \ot \out{i}{i}\big) V \nonumber \\
&=& \sum_{a,b} A^{\dagger}_{i,a}A_{i,b}\ot\out{a}{b},
\eea
where the unitarity of $V$ requires the operators $A_{i,a}$ to satisfy $\sum_i A_{i,a}^{\dagger}A_{i,b} = \delta_{a,b}\id_{S_0}$ and $\sum_a A_{i,a}A_{j,a}^{\dagger} = \delta_{i,j}\id_{S_0}$ with $A_{i,0}^{\dagger}A_{i,0} = E_i^{(0)}$. It should, however, be noted that the condition $\sum_i A_{i,a}^{\dagger}A_{i,b} = \delta_{a,b}\id_{S_0}$ is not required for $P_i$ (the latter expression) to be a projective measurement.
Also, when $A_{i,a}^{\dagger}A_{i,a} = E_i^{(a)}$, each $\bE_a = \{E_i^{(a)}\}$ becomes a POVM on $\cH^{S_0}$, and  $\bP$ is a Naimark extension for all $\bE_a$ because
\bea
\textmd{tr}\big\{P_i(\rho^{S_0}\ot\out{a}{a})\big\} &=& \sum_{b,c}\Big\{\iinner{c}{a}\iinner{a}{b}\textmd{tr}\big(A^{\dagger}_{i,b}A_{i,c}\rho^{S_0}\big)\Big\}\\
&=& \textmd{tr}\big(A^{\dagger}_{i,a}A_{i,a}\rho^{S_0}\big) = \textmd{tr}\big(E_i^{(a)}\rho^{S_0}\big).
\eea
It implies that
\bea
C(\rho^{S_0}\ot\out{a}{a},\bP) = C(\rho^{S_0},\bE_a)
\eea
for $a = 0,1,\cdots, n-1$ from the Naimark extension property and POVM-based coherence theory.
Then, by Theorem \ref{thm 2} and Eq.(\ref{eq:12}), we have
\bea
E_r^{S:A}\big(\Lambda_{SA}[(\rho^{S_0}\ot\out{a}{a}) \ot\out{0}{0}^A]\big) = C_r(\rho^{S_0},\bE_a),
\eea
where $\Lambda_{SA}[\rho^{SA}] = U\rho^{SA}U^\dag$ is a bipartite block incoherent operation with the unitary $U$ of Eq.(\ref{eq:incoherent unitary}) on $\cH^{SA} = \cH^{S}\ot\cH^A$ $(d_A\geq n)$.

Therefore, by attaching a selected ancilla $S_1$ amongst the reference incoherent states $\out{a}{a} \ (a = 0,1,\cdots, n-1)$ to the initial state $\rho^{S_0}$, we can selectively obtain the optimal entangled state via a bipartite block-incoherent operation $\Lambda_{SA}$. This has the advantage of being able to convert selected POVM-based coherence, among multiple POVMs, from the same initial state to entanglement by preparing only one bipartite block-incoherent operation and just changing the first auxiliary state without having to prepare a different incoherent operation for each POVM.

\section{Conclusion and summary}

In this paper, we investigated the quantitative relationship on the transformation of block coherence to quantum entanglement. We discussed some cases in generating the entanglement directly from POVM-based coherence.
Although it is constrained to directly generate entanglement from POVM-based coherence through POVM-based incoherent operations, we have presented strategies with some merit by linking it to block-coherence through a Naimark extension. We discovered that even if the initial state is taken in a relatively small space, such as a qubit, we can induce maximal entanglement in the extended space through the Naimark extension. This offers a great advantage in the preservation of quantum resources, given that the state of a smaller space is more favorable for storage and transmission.
Another is that one block-incoherent operation allows us to selectively induce entanglement based on two or more POVMs. Thus, we can selectively achieve the optimal entanglement state from the same initial state without having to prepare a different incoherent operation for each POVM.

\begin{acknowledgments}
This project is supported by the National Natural Science Foundation of China (Grants No. 12050410232, No. 12031004, and No. 61877054).
\end{acknowledgments}


%

\end{document}